\documentclass[runningheads,adpaper,12pt]{llncs}
\usepackage[utf8]{inputenc}
\usepackage{a4wide}
\usepackage{fullpage}
\usepackage{latexsym}
\usepackage{amssymb}
\usepackage{amsmath}
\usepackage{stmaryrd}
\usepackage[all, knot]{xy}
\usepackage{verbatim}
\usepackage[english]{babel}
\usepackage{graphicx}

\title{Decision Problem of Some Bundled Fragments}
\date{\today}

\author{Mo Liu}
\institute{Peking University}

\begin{document}

\maketitle
\section{Intro}

Bundled fragments bind first-order quantifiers and modalities together to be a new operator. The idea comes from \cite{Wang2017A} and it shows that one bundled fragment ``$\exists\Box$-fragment'' is decidable and has strong expressivity. As is known, fragments of first-order modal logic is mostly undecidable. So we are interested in the decision problem of other bundled fragments. The results by now are listed in the following table: \\

\centerline{
    \begin{tabular}{|c|c|c|}
\hline
     Fragment &Increasing Domain &Constant Domain  \\
\hline
     $\exists\Box$ & Decidable &Decidable\\
     &\cite{Wang2017A} &\cite{Padmanabha2018} \\
\hline
$\forall\Box$& Decidable & Undecidable \\
& \cite{Padmanabha2018} &\cite{Padmanabha2018} \\
\hline
$\Box\exists$& Decidable &  Unknown\\ &  &($\mathcal{L}_{\Box\exists^2}$ is Undecidable ) \\
\hline
$\Box\forall$ & Decidable & Undecidable \\
& &\cite{Padmanabha2018} \\
\hline
\end{tabular}
}

\section{Languages}
We have four bundled-fragments of first-order modal logic. 
\begin{definition}[$\mathcal{L}_{\exists\Box}$]
Given a set of predicates $\mathbf{P}$ and a set of variables $\mathbf{X}$, 
$$\varphi::= Px_1\cdots x_n\mid \neg\varphi\mid (\varphi\land\varphi)\mid \exists x\Box\varphi $$
where $P\in\mathbf{P}$ and $x_1,\cdots,x_n\in\mathbf{X}$
\end{definition}

\begin{definition}[$\mathcal{L}_{\forall\Box}$]
Given a set of predicates $\mathbf{P}$ and a set of variables $\mathbf{X}$, 
$$\varphi::= Px_1\cdots x_n\mid \neg\varphi\mid (\varphi\land\varphi)\mid \forall x\Box\varphi $$
where $P\in\mathbf{P}$ and $x_1,\cdots,x_n\in\mathbf{X}$
\end{definition}

\begin{definition}[$\mathcal{L}_{\Box\exists}$]
Given a set of predicates $\mathbf{P}$ and a set of variables $\mathbf{X}$, 
$$\varphi::= Px_1\cdots x_n\mid \neg\varphi\mid (\varphi\land\varphi)\mid \Box\exists x\varphi $$
where $P\in\mathbf{P}$ and $x_1,\cdots,x_n\in\mathbf{X}$
\end{definition}

\begin{definition}[$\mathcal{L}_{\Box\forall}$]
Given a set of predicates $\mathbf{P}$ and a set of variables $\mathbf{X}$, 
$$\varphi::= Px_1\cdots x_n\mid \neg\varphi\mid (\varphi\land\varphi)\mid \Box\forall x\varphi $$
where $P\in\mathbf{P}$ and $x_1,\cdots,x_n\in\mathbf{X}$
\end{definition}

\section{Semantics}
\begin{definition}[Increasing domain models]
An increasing domain model is a tuple $\mathcal{M}=(W,R,D,\delta,\{V_w\}_{w\in W})$, where 
\begin{itemize}
    \item $W$ is an non-empty set. 
    \item $R$ is a binary relation on $W$. 
    \item $D$ is an non-empty set. 
    \item $\delta: W\rightarrow \mathbf{2}^D$ is a function: for every $w\in W$ assigns a subset of $D$, and if $wRv$ then $\delta(w)\subseteq\delta(v)$.  
    \item For every $w\in W$, $V_w:\mathbf{P}\rightarrow \bigcup_{n\in\omega} \mathbf{2}^{D^n}$ if a function: for every $n$-predicate assigns a subset of $D^n$ as its interpretation. 
\end{itemize}
\end{definition}

$\delta(w)$ is the \textbf{local domain} of $w$, also noted as $D(w)$. if for any two states $w,v\in W$: $D(w)=D(v)$ holds, then $\mathcal{M}$ is a constant domain model. A \textbf{constant domain model} can be seen as $\mathcal{M}=(W,R,D, \{V_w\}_{w\in W}$). A \textbf{valuation} $\sigma$ is a function which for every variable $x\in \mathbf{X}$ assigns a element in $D$. 

\begin{definition}[Semantics]
Let $\mathcal{M}=(W,R,D,\delta,\{V_w\}_{w\in W})$ be a increasing domain model, for any $w\in W$ and any valuation $\sigma$,  
\begin{align*}
    \mathcal{M},w,\sigma &\vDash Px_1\cdots x_n \Leftrightarrow (\sigma(x_1),\cdots, \sigma(x_n))\in V_w(P) \\
    \mathcal{M},w,\sigma &\vDash\neg\varphi \Leftrightarrow \ \text{not} \ \mathcal{M},w,\sigma\vDash \varphi \\ 
    \mathcal{M},w,\sigma &\vDash\varphi\land\psi \Leftrightarrow \mathcal{M},w,\sigma \vDash\varphi \ \text{and} \ \mathcal{M},w,\sigma\vDash\psi\\ 
    \mathcal{M},w,\sigma &\vDash\exists x\Box\varphi \Leftrightarrow \ \text{there is a} \ d\in \delta(w) \ \text{such that for any} \ v\in W: \ \text{if} \ wRv, \ \text{then} \ \mathcal{M},v,\sigma(d/x)\vDash\varphi  \\
    \mathcal{M},v,\sigma &\vDash \forall x\Box\varphi \Leftrightarrow \ \text{for all} \  d\in\delta(w) \ \text{and for all} \ v\in W: \ \text{if} \ wRv \ \text{then} \ \mathcal{M},v,\sigma(d/x)\vDash\varphi \\
    \mathcal{M},w,\sigma &\vDash \Box\exists x\varphi \Leftrightarrow \ \text{for all} \ v\in W \ \text{with} \ wRv, \ \text{there is a} \ d\in\delta(v) \ \text{ such that} \ \mathcal{M},v,\sigma(d/x)\vDash\varphi \\
    \mathcal{M},w,\sigma &\vDash \Box\forall x\varphi\Leftrightarrow \ \text{for all} \ v\in W \ \text{with} \ wRv \ \text{ and for all} \ d\in \delta(v): \mathcal{M},v,\sigma(d/x)\vDash\varphi
\end{align*}

\end{definition}

\section{Decidability results}
For the convenience of proving, we introduce positive normal form (PNF) and bind two bundled-fragment together. 

\begin{definition}[$\mathcal{M}_{\exists\Box\forall\Box}$-PNF]
$$\varphi::= P\overline{x}\mid \neg P\overline{x}\mid(\varphi\land\varphi)\mid(\varphi\lor\varphi)\mid\exists x\Box\varphi\mid\exists x\Diamond \varphi\mid\forall x\Box\varphi \mid\forall x\Diamond\varphi$$
\end{definition}

\begin{definition}[$\mathcal{M}_{\Box\exists\Box\forall }$-PNF]
$$\varphi::= P\overline{x}\mid \neg P\overline{x}\mid(\varphi\land\varphi)\mid(\varphi\lor\varphi)\mid\Box\exists x\varphi\mid\Diamond\exists x\varphi\mid\Box\forall x\varphi \mid\Diamond \forall x\varphi$$
\end{definition}

Clearly, every bundled-fragment FOML formula can be rewritten into an equivalent formula in PNF. We call a formula \textbf{clean} if no variable occurs both bound and free in it and every use of a quantifier quantifies a distinct variable. It's not hard to see that we can re We call formulas in forms $P\overline{x}$ and $\neg P\overline{x}$ as \textit{literals}, and we use ``\textit{lit}'' to denote the set of all literals. 

\begin{definition}[Tableau]

 A tableau is a tree structure $T=(W,V,E,\lambda)$ where $W$ is a finite set, $(V,E)$ is a rooted tree and $\lambda:V\rightarrow L$ is a labelling map. Each element in $L$ is of the form $(w:\Gamma,\sigma)$, where $w\in W$, $\Gamma$ is finite set of formulas and $\sigma$ is a partial identity mapping on $\mathbf{X}$. Then intended meaning of the label is that the node constitutes a world $w$ that satisfies the formulas in $\Gamma$ with the assignment $\sigma$. 
 \end{definition}

\begin{theorem}
The Satisfiability problem of $\mathcal{L}_{\exists\Box\forall\Box}$ fragment over increasing domain models is decidable. 
\end{theorem}
\begin{proof}
See \cite{Padmanabha2018}
\end{proof}
Following the method in \cite{Padmanabha2018}, we can show $\Box\exists$-fragment and $\Box\forall$-fragment are also  decidable by modifying some tableau rules. 

\begin{definition}
Tableau rules for increasing domain models for the $\mathcal{L}_{\Box\exists\Box\forall}$ fragment are given by: 
\\ 

\begin{tabular}{|c c|}
\hline
\multicolumn{2}{|c|}{} \\
     $\dfrac{w:\phi_1\lor\phi_2,\Gamma,\sigma}{w: \phi_1,\Gamma,\sigma\mid w:\phi_2,\Gamma,\sigma}\ (\lor)$ &  $\dfrac{w:\phi_1\land\phi_2,\Gamma,\sigma}{w:\phi_1,\phi_2,\Gamma,\sigma} \ (\land)$ \\
    \multicolumn{2}{|c|}{} \\
     \hline
     \hline
      \multicolumn{2}{|c|}{Given $n_1,n_2,s\geq 0$ and $m_1,m_2\geq 1$:} \\
      \multicolumn{2}{|c|}{}\\
      \multicolumn{2}{|c|}{$\cfrac{\genfrac{}{}{0pt}{0}{w:\Box\exists x_1\phi_1,\cdots,\Box\exists x_{n_1}\phi_{n_1},\Box\forall y_1\psi_1,\cdots,\Box\forall y_{n_2}\psi_{n_2},}{\Diamond\forall z_1\alpha_1,\cdots,\Diamond \forall z_{m_1}\alpha_{m_1},\Diamond\exists u_1\beta_1,\cdots,\Diamond\exists u_{m_2}\beta_{m_2},} \ r_1,\cdots,r_s,\sigma}{\genfrac{}{}{0pt}{0}{\{(wv_{z_i}:\phi_1,\cdots,\phi_{n_1},\{\psi_j^*[x/y_j]\mid j\in[1,n_2],x\in\sigma'\},\{\alpha_i[y/z_i]\mid y\in\sigma'\},\sigma')\mid i\in [1,m_1]\}\bigcup }{\{(wv_{u_k}:  \phi_1,\cdots,\phi_{n_1},\{\psi_j^*[x/y_j]\mid j\in[1,n_2],x\in\sigma'\},\beta_k,\sigma')\mid k\in[1,m_2]\}}} \quad (BR)$} \\ 
      \multicolumn{2}{|c|}{} \\ 
      \hline
      \hline
      \multicolumn{2}{|c|}{Given $n_1,n_2\geq 1, s\geq 0$:} \\
      \multicolumn{2}{|c|}{} \\
      \multicolumn{2}{|c|}{$\dfrac{w:\Box\exists x_1\phi_1,\cdots,\Box\exists x_{n_1}\phi_{n_1},\Box\forall y_1\psi_1,\cdots,\Box\forall y_{n_2}\psi_{n_2},r_1,\cdots,r_k,\sigma}{w:r_1,\cdots,r_k,\sigma} \quad (END)$} \\
      \multicolumn{2}{|c|}{} \\
      \hline
     \multicolumn{2}{c}{where $\sigma'=\sigma\cup\{(x_i,x_i)\mid i\in[1,n_1]\}\cup\{u_j\mid j\in [1,m_2]\}$ and $r_1,\cdots,r_s\in lit$}
\end{tabular}
\end{definition}

\begin{theorem}
For any clean $\mathcal{L}_{\Box\exists\Box\forall}$ PNF formula $\theta$, there is an open tableau with the root $(r:\{\theta\},\sigma_r)$ where $dom(\sigma_g)=\{x\mid x \ \text{is free in} \ \theta \cup \{z\}$ ($z$ is a variable doesn't occur in $\theta$), if and only if $\theta$ is satisfiable in an increasing domain model. 
\end{theorem}
\begin{proof}
Let $T$ be any tableau starting from $(r:\{\theta\},\sigma_r)$. For any node on $T$ with the label $(w:\Gamma,\sigma)$. We can prove ``\textbf{(C)}: $\bigwedge\Gamma$ is clean. '' by induction on $T$ from the root. The method is similar with the proof in Theorem 1. \\ 

We define an increasing model $\mathcal{M}=(W,D,\delta,R,\{V_w\}_{w\in W})$, where 
\begin{itemize}
    \item $W=\{w\mid (w:\Gamma,\sigma)\in \Gamma \text{for some} \ \Gamma,\sigma \}$. 
    \item $\delta(w_=Dom(\sigma_w)$
    \item $D=\bigcup_{w\in W}\delta(w)$ 
    \item $wRv$ iff there is a $v'$ such that $v=wv'$
    \item $\overline{x}\in V_w(P)$ iff $P\overline{x}\in\Gamma$ 
\end{itemize}
As $g\in W$ and $z\in D$, $W$ and $D$ are non-empty. If $wRv$, then $dom(\sigma_w)\subseteq dom(\sigma_v)$. Since $T$ is an open tableau, this model is well-defined. We show that $\mathcal{M},r,\sigma$ is a model for $\theta$. We make an induction on $T$ from leaf nodes to show that for all nodes $(w: g,\sigma)$ in $T$, $\mathcal{M},w,\sigma\vDash\Gamma$.  \\
\begin{itemize}
    \item For all leaf nodes $(w: r_1,\cdots, r_k,\sigma)$, there are only literals in their labels. By the definition of $V$, for every literal $r_i (i\in [1,k]$), we have $r_i\in \Gamma$ iff $\mathcal{M},w,\sigma\vDash r_i$. 
    \item For the nodes by using $(\land)$ and $(\lor)$ rules, it's straightforward by IH. 
    \item For the nodes by using the $(END)$ rule, by IH $\mathcal{M},w,\sigma\vDash r_1\land\cdots\land r_k$, and as $w$ has no successor, we have $\mathcal{M},w,\sigma\vDash \Box\exists x_1\varphi_1\land\cdots\land\Box\exists x_{n_1}\varphi_{n_1}\land \Box\forall y_1\psi_1\land\cdots\land \Box\forall y_{n_2}\psi_{n_2}\land r_1\land\cdots\land r_k$. 
    \item For any node by using $(BR)$ rule, suppose its label is $(w:\Gamma,\sigma)$, and \begin{multline*}
    \Gamma= \{\Box\exists x_i\phi_i\mid i\in [1,n_1]\}\cup\{\Box\forall y_j\psi_j\mid j\in[1,n_2]\} \ \cup \\ 
    \{\Diamond\forall z_k\alpha_k\mid k\in [1,m_1]\}\cup\{\Diamond\exists u_l\beta_l\mid l\in[1,m_2]\}\cup\{r_1,\cdots,r_s\}
    \end{multline*} 
    By IH, for any $i\in [1,n_1]$ $\mathcal{M},wv_{z_i},\sigma'\vDash\{\varphi_1,\cdots,\varphi_{n_1},\{\psi_j[x/y_j]\mid j\in[1,n_2],x\in\sigma'\},\{\alpha_i[y/z_i]\mid y\in\sigma'\}$ and for any $l\in [1,m_2]$, $\mathcal{M},wv_{u_l},\sigma'\vDash\{\varphi_1,\cdots,\varphi_{n_1},\{\psi[x/j_j]\mid j\in[1,n_2], x\in dom(\sigma')\},\beta_l$. We need to show $\mathcal{M},w,\sigma\vDash\bigwedge\Gamma$. For literals in $\Gamma$, it's similar to the case of leaf nodes. \\
    
    For $\Box\exists x_i\varphi_i (i\in[1,n_1])$, given any successor $wv_{z_k}$ or $wv_{u_l}$ of $w$ let $x_i$ be the witness. Obviously, $x_i \in dom(\sigma)=\delta(wv_{z_k})=\delta(wv_{u_l})$. Note that the difference between $\sigma$ and $\sigma'$ is $\{(x_i,x_i)\mid i\in [1,n_i]\}$ and $\{(u_l,u_l)\mid l\in[1,m_2]$. By (C), for all $a\neq i$, $x_a$ doesn't occur in $\varphi_i$ and any $u_l (l\in[1,m_2])$ doesn't occur in $\varphi_1$. By Ih, $\mathcal{M},wv_{z_k},\sigma'\vDash\varphi_i$, thus $\mathcal{M},wv_{z_k},\sigma(x_i/x_i)\vDash\varphi_i4$; Similarly, $\mathcal{M}.wv_{u_r},\sigma(x_i,x_i)\vDash\varphi$. Therefore, $\mathcal{M},w,\sigma\vDash\Box\exists x_i\varphi_i$. \\
    
    For $\Box\forall y_j\psi_j\ (j\in[1,n_2])$, every $w$'s successor $wv_{z_k}$ or $wv_{u_r}$, $dom(\sigma')=\delta(wv_{z_k})=\delta(wv_{u_r})$. For any element $x\in dom(\sigma')$, by IH $\mathcal{M},wv_{z_k},\sigma'\vDash\phi_j[x/y_j]$ and $\mathcal{M},wv_{u_r},\sigma\vDash\psi_j[x/y_j]$. By (C),$\bigwedge\Gamma$ is clean, thus $x$ is sub-free with respect to $y_j$ in $\psi_j$. Thus $\mathcal{M},wv_{z_k},\sigma'(x/y_j)\vDash\psi_j$ and $\mathcal{M},wv_{u_r},\sigma'(x/y_j)\vDash\psi_j$. By (C) again, $\sigma'(x/y_j)$ and $\sigma(x/y_j)$ are agree with respect to $\psi_j$, thus $\mathcal{M},wv_{z_k},\sigma(x/y_j)\vDash\psi_j$ and $\mathcal{M},wv_{u_r},\sigma(x/y_j)\vDash\psi_j$. Since $x$ is arbitrary, we have $\mathcal{M},w,\sigma\vDash\Box\forall y_j \psi_j$. \\
    
    For $\Diamond\forall z_k\alpha_k\ (k\in[1,m_1])$, by the definition of $R$, $wv_{z_k}$ is $w$'s successor. For any $y\in\delta(wv_{z_k})=Dom(\sigma')$, by IH we have $\mathcal{M},wv_{z_k},\sigma'\vDash \alpha_k[y/z_k]$. As $\sigma'(y)=y$, and $y$ is sub-free with respect to $z_k$ in $\alpha_k$, thus  $\mathcal{M},wv_{y_j},\sigma'(y/z_k)\vDash\alpha_k$. As mentioned before, $\sigma'$ and  $\sigma$ differ in $\{(x_i,x_i)\mid i\in[1,n_1]\}$ and $\{(u_l,u_l)\mid l\in [1,m_2]\}$. By (C), $\bigwedge\Gamma$ is clean, any $x_i$ or $u_l$ doesn't occur freely in $\alpha_k$, thus $\sigma'$ and $\sigma$ are agree with respect to $\alpha_k$. Thus $\sigma'(y/z_k)$ and $\sigma(y/z_k)$ are agree with respect to $\alpha_k$. Therefore, $\mathcal{M},wv_{z_k},\sigma(y/z_k)\vDash \alpha_k$, and then $\mathcal{M},w,\sigma\vDash\Diamond\forall z_k\alpha_k$.\\
      
    For $\Diamond\exists u_r\beta_r\ (r\in[1,m_2])$,by the definition of $R$, $wv_{u_l}$ is $w$'s successor. By IH, $\mathcal{M},wv_{u_r},\sigma'\vDash\beta_l$ and $u_l\in dom(\sigma')$, thus $\mathcal{M},wv_{u_l},\sigma'\vDash\exists u_l\beta_l$, by $\bigwedge\Gamma$ is clean, we have $\mathcal{M},wv_{u_l},\sigma\vDash\exists u_l\beta_l$, and then $\mathcal{M},w,\sigma\vDash\Diamond\exists u_l\beta_l$. 
    
\end{itemize}
On the other hand, we need to show that all the tableau rules preserve satisfiability so that we can construct an open tableau from $(r:\{\theta\},\sigma_r)$. For the $(END)$,  $(\land)$ and $(\lor)$ rules, it's obvious that they preserve satisfiablity.  For the $(BR)$ rule, let 
\begin{multline*}
    \Gamma= \{\Box\exists x_i\phi_i\mid i\in [1,n_1]\}\cup\{\Box\forall y_j\psi_j\mid j\in[1,n_2]\} \ \cup \\ 
    \{\Diamond\forall z_k\alpha_k\mid k\in [1,m_1]\}\cup\{\Diamond\exists u_r\beta_r\mid r\in[1,m_2]\}\cup\{l_1,\cdots,l_s\}
    \end{multline*}
is satisfialbe, we show that 
\begin{itemize}
    \item[A] For any $i\in[1,m_1]$, $\{\phi_1,\cdots,\phi_{n_1},\{\psi_j[x/y_j]\mid j\in[1,n_2],x\in\sigma'\},\{\alpha_i[y/z_i]\mid y\in\sigma'\}\}$ are satisfiable. 
    \item[B] For any $r\in [1,m_2]$, $\{\phi_1,\cdots,\phi_{n_1},\{\psi_j[x/y_j]\mid j\in[1,n_2],x\in\sigma'\},\beta_r \}$ are satisfiable.
\end{itemize}
Let $\mathcal{M}=(W,D,R,\delta,\{V_w\}_{w\in W})$ be an increasing domain model, $w\in W$,  $\eta$ is a valuation that for any free variable $x$ in $\Gamma$: $\eta(x)\in \delta(w)$ and $\mathcal{M},w,\eta\vDash\bigwedge\Gamma$. As $\mathcal{M},w,\eta\vDash\bigwedge\{\Box\exists x_i\phi_i\mid i\in[1,n_1]\}$, by semantics, for $w$'s every successor $v$, there is an $a_i\in\delta(v)$ such that $\mathcal{M},v,\eta(a_i/x_i)\vDash\phi_i$\ $(i\in [1,n_1])$. Since $\phi_i$ is clean, we have $\mathcal{M},v,\eta(\overline{a}/\overline{x})\vDash\phi_1\land\cdots\land\phi_{n_1}$. As $\mathcal{M},w,\eta\vDash\bigwedge\{\Box\forall y_j\psi_j\mid j\in[1,n_2]\}$, for any $w$'s successor $v$, there is a $b \in\delta(v)$ such that $\mathcal{M},v,\eta(b/y_j)\vDash\psi_j$. Note that for all $i\in[1,n_1]$, $y_j\neq x_i$, thus $y_j$ doesn't occur in $\phi_1,\cdots,\phi_{n_1}$, and then  $\mathcal{M},v,\eta(\overline{a}/\overline{x})(b/y_j)\vDash\phi_1\land\cdots\land\phi_{n_1}\land\psi_j$. For any $y\in dom(\sigma')$ and $\eta(y)=b$ (by the range of $\eta$, we know that $b\in \delta(v)$), by (C), we have $y$ is sub-free with respect to $y_j$ in $\psi_j$, thus $\mathcal{M},v,\eta(\overline{a}/\overline{x})\vDash\phi_1\land\cdots\land\phi_{n_1}\land\psi_j[y/y_j]$ for every $y\in dom(\sigma')$. Therefore,  $\mathcal{M},v,\eta(\overline{a}/\overline{x})\vDash\phi_1\land\cdots\land\phi_n\land\bigwedge\{\psi[x/y_j]\mid x\in dom(\sigma')\}$. By the above results, we can show in two parts, for the nodes using by $(BR)$ as $wv_{z_i} (i\in [1,m_1])$ and $wv_{u_r} (r\in [1,m_2])$, the formulas in their labels are all satisfiable.  \\

\par A: By $\mathcal{M},w,\eta\vDash\bigwedge\{\Diamond\forall z_k \alpha_k\mid k\in[1,m_1]\}$, we know that for every $i\in [1,m_1]$, there is a $w$'s successor  $v_i$ such that $\mathcal{M},v_i,\eta\vDash\forall z_k\alpha_k$. For any $y\in dom(\sigma')$, $\eta(y)\in \delta(w)$, and since $\mathcal{M}$ is an increasing domain model, $\delta(w)\subseteq\delta(v_i)$, thus $\eta(y)\in\delta(v_i)$, and then $\mathcal{M},v_i,\eta\vDash\bigwedge\{\alpha_i[y/z_i]\mid y\in dom(\sigma')\}$. As for any $i\in[1,m_1]$, any $x\in dom(\sigma')$ doesn't occur freely in $\alpha_i$, thus $\mathcal{M},v_i,\eta(\overline{a}/\overline{x})\vDash\bigwedge\{\phi_1,\cdots,\phi_{n_1},\{\psi_j[x/y_j]\mid j\in[1,n_2],x\in\sigma'\},\{\alpha_i[y/z_i]\mid y\in\sigma'\}\}$, and then $\{\phi_1,\cdots,\phi_{n_1},\{\psi_j[x/y_j]\mid j\in[1,n_2],x\in\sigma'\},\{\alpha_i[y/z_i]\mid y\in\sigma'\}\}$ is satisfiable. 
\par B: Similarly, by $\mathcal{M},w,\eta\vDash\bigwedge\{\Diamond\exists u_l\beta_l\mid l\in[1,m_2]\}$, we have for any $l\in[1,m_2]$, there is a $w$'s successor $v_l$ and an $d\in \delta(v_l)$ such that $\mathcal{M},v_l,\eta(d/u_l)\vDash\beta_l$. Since $u_l$ doesn't occur freely in $\phi_1\land\cdots\land\phi_{n_1}\land\bigwedge\{\psi[x/y_j]\mid x\in dom(\sigma')\}$, $\mathcal{M},v_l,\eta(\overline{a}/\overline{x})(d/u_l)\vDash\bigwedge\{ \phi_1,\cdots,\phi_{n_1},\{\psi_j[x/y_j]\mid j\in[1,n_2],x\in\sigma'\},\beta_l  \}$, that is to say $\{ \phi_1,\cdots,\phi_{n_1},\{\psi_j[x/y_j]\mid j\in[1,n_2],x\in\sigma'\},\beta_l\}$ is satisfiable. 
\end{proof}

\begin{corollary}
The satisfiability problem of $\mathcal{L}_{\Box\exists\Box\forall}$ fragment over increasing domain models is decidable. 
\end{corollary}

For satisfiability of all bundled-fragment over increasing domain models, since we can add variables of which the quantifiers are eliminated by $(BR)$ rule into domains of successors, the tableau method works very well. But how about constant domain models? In this case, we have to set a universal domain at the start of the tableau construction and use only these elements as witness. 

We have a decidable result for $\mathcal{L}_{\exists\Box}$ fragment over constant domain models. We do this by calculating precise bound on how many new elements need to be added for each sub-formula of the form $\exists x\Box\varphi$ and include as many as needed at the beginning of the tableau construction.

\begin{theorem}
The satisfiability problem for $\mathcal{L}_{\exists\Box}$-formulas over constant domain models is decidable. 
\end{theorem}
\begin{proof}
See \cite{Padmanabha2018}
\end{proof}

\section{Undecidability Results}
We prove that the satisfiability problem for the $\mathcal{L}_{\forall\Box}$ fragment over the class of constant domain models is undecidable even when the atomic predicates are restricted to be unary. \\

Consider $\mathcal{L}_{FO(R)}$, the fist order logic with only variables as terms and no equality, and the single binary relation $R$. We know $\mathcal{L}_{FO(R)}$ is undecidable, and we reduce $\mathcal{_L}_{FO(R)}$ into $\mathcal{L}_{\Box\forall}$. \\

For any quantifier-free $\mathcal{L}_{FO(R)}$ formula, we define the translation $Tr$ as 
\begin{itemize}
    \item $Tr(Rxy):=\exists z\Diamond(Px\land Qy)$, where $z$ is distinct from $x$ and $y$
    \item $Tr(\neg\beta):=\neg Tr(\beta)$
    \item $Tr(\beta_1\land\beta_2):= Tr(\beta_1)\land Tr(\beta_2)$.
\end{itemize}
Now consider $\mathcal{L}_{FO(R)}$ formulas in prenex normal forms. Let $\alpha$ be a sentence as $Q_1x_1Q_2x_2\cdots Q_nx_n \beta$ where $\beta$ is quantifier-free. We define: 
$$\psi_\alpha:= Q_1x_1\Delta_1 Q_2x_2\Delta_2\cdots Q_nx_n\Delta_n Tr(\beta)$$
where $Q_ix_i\Delta_i:=\exists_i x_i\Diamond$ when $Q_i=\exists$ and $Q_ix_i\Delta_i=\forall_i x_i \Box$ when $Q_i=\forall$.  \\ 

In addition, we need to make the model has a certain depth correponding to $\alpha$'s quantifier so that there is no dead end state before the modal depth of its translation. We use formula $\lambda_n:= \bigwedge^n_{j=0}(\forall z\Box)^j(\exists z\Diamond \top)$ to ensure for all state at depth $i\leq n$ it has a successor.  \\

Finally, we need to ensure that $\exists z\Diamond(Px\land Qy)$ is evaluated uniformly at the ``tail'' states. We have $$\gamma_n:=\forall z_1\Box\forall z_2\Box((\exists z\Diamond)^n(\exists z\Diamond(Pz_1\land Qz_2))\rightarrow (\forall z\Box)^n(\exists z\Diamond(Pz_1\land Qz_2)))$$
where $z,z_1,z_2$ doesn't appear in $\alpha$. 

\begin{definition}
Given a $\mathcal{L}_{FO(R)}$ sentence $\alpha:= Q_1x_1Q_2x_2\cdots Q_nx_n\beta$ in prenex normal form, the translation $\mathcal{L}_{\forall\Box}$ formula $\varphi_\alpha$ is given by: $$\varphi_\alpha:=(\forall z\Box)^2 \psi_\alpha\land \lambda_{n+2}\land\gamma_n$$
where $z$ doesn't appear $\alpha$
\end{definition}

Note that for any $\mathcal{L}_{FO(R)}$ sentence $\alpha$ of quantifier depth $n$, we get a translated formula $\varphi_\alpha$ of modal depth $n+3$. 

\begin{theorem}
For any $\mathcal{L}_{FO(R)}$ sentence $\alpha:=Q_1x_1Q_2x_2\cdots Q_nx_n\beta$ in prenex normal form, $\alpha$ is satisfiable iff $\varphi_\alpha$ is satisfiable on constant domain model. 
\end{theorem}
\begin{proof}
See \cite{Padmanabha2018}
\end{proof}
\begin{corollary}
The satisfiability problem for the $\mathcal{L}_{\Box\forall}$ fragment over constant domain models is undecidable. 
\end{corollary}
 But we cannot prove $\mathcal{L}_{\Box\exists}$ is undeciable by this method. The reason is that for  
$$\gamma_n:=\forall z_1\Box\forall z_2\Box((\exists z\Diamond)^n(\exists z\Diamond(Pz_1\land Qz_2))\rightarrow (\forall z\Box)^n(\exists z\Diamond(Pz_1\land Qz_2)))$$
The version of $\mathcal{L}_{\Box\exists}$ is
$$\gamma_n':=\Box\forall z_1\Box\forall z_2 ((\Diamond\forall z)^n(\Diamond \forall z( Pz_1\land Qz_2))\rightarrow (\Box\exists z)^n(\Diamond\forall z(Pz_1\land Qz_2)$$
$\gamma_n'$ cannot ensure that $\Diamond \forall z(Px\land Qy)$ is evaluated uniformly at the ``tail'' states. But when we consider a $\mathcal{L}_{\Box\exists^2}$ fragment which bind one modality and two quantifiers together, we can have a formula 
$$\gamma_n^*:=\Box\forall z_1\forall z_2 ((\Diamond\forall z)^n(\Diamond \forall z( Pz_1\land Qz_2))\rightarrow (\Box\exists z)^n(\Diamond\forall z(Pz_1\land Qz_2)$$
The other parts of translation is similar:
\begin{itemize}
    \item $Tr(\beta):=\Diamond\forall z\forall z'(Px\land Qy)$ where $z,z'$ is different from $x,y$
    \item $Tr(\neg\beta):=\neg Tr(\beta)$
    \item $Tr(\beta_1\land\beta_2):= Tr(\beta_1)\land Tr(\beta_2)$
    \item $\psi^*_\alpha:= \Delta_1 Q_1 x_1 Q_1 x_1' \cdots \Delta_nQ_nx_nQ_nx_n' Tr(\beta)$. 
    \item $\lambda^*_n:= \bigwedge^n_{j=0}(\Box\exists z\exists z')^j\Diamond\forall z\forall z'\top $
\end{itemize}
\begin{definition}
Given a $\mathcal{L}_{FO(R)}$ sentence $\alpha:= Q_1x_1Q_2x_2\cdots Q_nx_n\beta$ in prenex normal form, the translation $\mathcal{L}_{\Box\exists^2}$ formula $\varphi_\alpha$ is given by: $$\varphi_\alpha:=\Box\exists z\exists z' \psi^*_\alpha\land \lambda^*_{n+1}\land\gamma_n^*$$
where $z,z'$ doesn't appear in $\alpha$
\end{definition}

\begin{theorem}
For any $\mathcal{L}_{FO(R)}$ sentence $\alpha:=Q_1x_1Q_2x_2\cdots Q_nx_n\beta$ in prenex normal form, $\alpha$ is satisfiable iff $\varphi_\alpha^*$ is satisfiable on constant domain model. 
\end{theorem}

\bibliography{main}

\end{document}